\documentclass{article}
\usepackage{graphicx}
\usepackage{amsfonts}
\usepackage{amsmath}
\usepackage{amssymb}
\usepackage{url}

\usepackage{dsfont}
\usepackage{fancyhdr}
\usepackage{indentfirst}
\usepackage{enumerate}
\usepackage[normalem]{ulem}
\usepackage{mathrsfs}
\usepackage{multirow}
\usepackage{diagbox}

\usepackage[T1]{fontenc}

\usepackage[colorlinks=true,citecolor=blue]{hyperref}
\usepackage{amsthm}
\usepackage{color}

\definecolor{DarkGreen}{rgb}{0.2,0.6,0.2}

\definecolor{purple}{rgb}{0.6,0.3,0.8}

\usepackage{natbib}
\usepackage{comment}

\def\d{\mathrm{d}}
\def\laweq{\buildrel \mathrm{d} \over =}

\newcommand{\E}{\mathbb{E}}

\newcommand{\R}{\mathbb{R}}

\newcommand{\N}{\mathbb{N}}
\newcommand{\p}{\mathbb{P}}
\newcommand{\X}{\mathcal{X}}

\renewcommand{\L}{\mathcal{L}}

\newcommand{\id}{\mathds{1}}

\renewcommand{\ge}{\geqslant}
\renewcommand{\le}{\leqslant}
\renewcommand{\geq}{\geqslant}
\renewcommand{\leq}{\leqslant}
\renewcommand{\epsilon}{\varepsilon}

\theoremstyle{plain}
\newtheorem{theorem}{Theorem}[section]
\newtheorem{corollary}[theorem]{Corollary}
\newtheorem{lemma}[theorem]{Lemma}
\newtheorem{proposition}[theorem]{Proposition}
\theoremstyle{definition} 

\newtheorem{example}[theorem]{Example}

\theoremstyle{remark}
\newtheorem{remark}[theorem]{Remark}

\numberwithin{equation}{section} %\numberwithin{theorem}{section}

\newcommand{\VaR}{\mathrm{VaR}}

\newcommand{\ES}{\mathrm{ES}}

\usepackage{setspace}
\usepackage{footmisc}
\title{Coherent risk measures and uniform integrability} 
\author{Muqiao Huang\thanks%
  {Department of Statistics and Actuarial Science,
  University of Waterloo, Canada.
 \href{mailto:m5huang@uwaterloo.ca}{m5huang@uwaterloo.ca}.} \and Ruodu Wang\thanks%
  {Department of Statistics and Actuarial Science,
  University of Waterloo, Canada.
 \href{mailto:wang@uwaterloo.ca}{wang@uwaterloo.ca}.}}
\begin{document}

\maketitle

\begin{abstract}  
 
We establish a profound connection between coherent risk measures, a prominent object in quantitative finance, and uniform integrability, a fundamental concept in probability theory.
Instead of working with absolute values of random variables, which is convenient in studying integrability, we work directly with random losses and gains, 
which have a clear financial interpretation. 
We introduce a technical tool called the folding score of distortion risk measures. The analysis of the folding score allows us to convert some conditions on absolute values to those on losses and gains.  
%The folding score turns out to be finite unless the distortion risk measure is the mean. 
As our main results, we obtain three sets of equivalent conditions for uniform integrability. In particular,
a set is uniformly integrable if and only if one can find a coherent distortion risk measure that is bounded on the set, but not finite on $L^1$. 


%        \textbf{Mathematics Subject Classification (2010)}: 91G70
        
%        \textbf{JEL Classification}: C65

\end{abstract}

\section{Introduction}
 
Coherent risk measures, introduced by \cite{ADEH99}, have been a cornerstone of quantitative finance.
The development on risk measures has proven prolific in both academic research and  financial regulation; we refer to \cite{FS16} and \cite{MFE15} for general treatments on risk measures and their applications. In particular, the Value-at-Risk (VaR) and the Expected Shortfall (ES, also known as CVaR, TVaR and AVaR) are the two most important risk measures in the financial industry; see the regulatory document of \cite{BASEL19} in the banking sector. 

Our main goal is to connect coherent risk measures to uniform integrability, an old concept in probability theory, useful in many domains of analysis. As an example of financial relevance,  for a sequence of random variables that converges in distribution, e.g., estimated financial models from data that converge to the true model, uniform integrability guarantees convergence of many risk measures, including ES.
Continuity of a risk measure with respect to distributional convergence is associated with robustness of the risk measure by 
\cite{CDS10}.  See \cite{KSZ14} and \cite{ESW22} for recent developments on the robustness of risk measures. 
Uniform integrability is a useful condition in many financial models, especially in the context of martingales; see e.g., \cite{KOZ18} and \cite{BDD20}.

We will study conditions 
on values of a coherent risk measure applied to random variables in a set $\mathcal S$ that characterize 
uniform integrability of $\mathcal S$.  
There is a subtle difference for conditions typically used in probability theory and those in the literature of risk measures. 
Due to the definition of uniform integrability, 
it is conventional to consider conditions on $|X|$ for $X\in \mathcal S$.
Suppose that the random variable $X$  represents a financial position (an asset), with its positive realized values representing losses and negative ones representing gains, a convention used by \cite{MFE15} (a sign change from \cite{FS16}).
For a risk measure $\rho$, the value $\rho(X)$ has a concrete interpretation as regulatory capital requirement for a long position in the asset.
Similarly, $\rho(-X)$ also has a concrete meaning, as the regulatory capital requirement for a short position in the asset.
In conic finance (see \cite{MC10}),  $\rho(-X)$ can also represent the price of the asset with payoff $-X$. 
In an insurance setting, $\rho(X)$   would represent the price of the random loss  $X$.  
On the other hand, $\rho(|X|)$, which evaluates the risk of the absolute value of $X$, does not have a reasonable financial meaning. Recall that monotonicity of $\rho $ has the natural  interpretation of ``less loss is better", 
but $X\mapsto \rho(|X|)$ is not monotone and loses such meaning. 
The quantity $\rho(|X|)$ seems to be only relevant for technical analysis but it does not appear in any economic problems.

As a concrete example, in the optimal investment problem in Section \ref{sec:investment}, an upper bound on $\rho(X)$ 
naturally appears as a risk constraint (see e.g., \cite{BS01} and \cite{RU02} for using risk measures as risk constraints),
and an upper bound on $\rho(-X)$ 
naturally appears as a budget constraint (one can use a different $\rho$ when both constraints exist). 
In such a context, $\rho(|X|)$ is irrelevant.

For the above reasons, we would hope to formulate conditions using $\rho(X)$ and $\rho(-X)$, which are financially interpretable, instead of using $\rho(|X|)$.  Moreover, conditions on $\rho(X)$ and $\rho(-X)$ are easier to check for common financial models of $X$ than $|X|$; see \citet[Chapter 2]{MFE15} for some existing formulas.
To formulate conditions on $\rho(X)$ and $\rho(-X)$, we introduce a technical tool called the folding score of distortion risk measures in Section \ref{sec:general}, which quantifies the supremum of $\rho(|Y|)/\max\{\rho(Y),\rho(-Y)\}$ over $Y\in L^1$.
A useful result on the folding score (Theorem \ref{th:main}) allows us to conveniently  convert conditions on $\rho(|X|)$ to those on $\rho(X)$ and $\rho(-X)$. In particular, we show that for the class of coherent distortion risk measures, boundedness of $\rho(|X|)$ for $X\in \mathcal S$
is equivalent to 
 boundedness of $\rho(X)$
 and $\rho(-X)$, unless $\rho$ is the expectation. 

 With the help from the folding score, we obtain in Section \ref{sec:UI} three different sets of equivalent conditions for uniform integrability, using ES, coherent distortion risk measures, or law-invariant coherent risk measures. 
As a particularly convenient result, in Theorem \ref{th:3} we find that 
 $\mathcal S$ is uniformly integrable if and only if one can find a coherent distortion risk measure that is not finite on $L^1$  and bounded on $ \mathcal S$ and $-\mathcal S$.
This result closely resembles the classic characterization of uniform integrability via the de la Vall\'ee Poussin
criterion; see e.g., \citet[Theorem T22]{M66}. 
As intermediate results, we also obtain that a law-invariant coherent risk measure is finite on $L^1$ if and only if it is controlled by a constant times the expectation for positive random variables, and for coherent distortion risk measures this is equivalent to a bounded slope of its distortion function. These intermediate results are closely related to the recent results of  \cite{FW22} in the context of rearrangement-invariant Banach norms.
Further, we extend the  boundedness condition of $\rho$ on $\mathcal S$ and $-\mathcal S$ discussed above to a similar condition on two possibly different risk measures in Theorem \ref{th:3-ch}.

Three consequences of our main results on a law of large numbers, the convergence of ES values, and convergent sequences in $1$-Wasserstein distance are presented in Section \ref{sec:application}, and they direct follow from uniform integrability.  As a financial application, in  Section \ref{sec:investment} we consider  an investment problem, where a decision maker has an optimization problem subject to a risk constraint and a price constraint on the financial position, via two coherent risk measures. 
Applying  Theorem \ref{th:3-ch},  the two constraints imply uniform integrability of the set of possible positions, and this fact can be further leveraged to show the convergence of approximate optimizers to a true optimizer. 
% As a financial application, a simple investment optimization problem in a complete market is presented in Section \ref{sec:investment}, where our main results are useful for showing  the convergence of approximate optimizers.
Section \ref{sec:6} concludes the paper.

\section{Notation and preliminaries} \label{sec:2}

Let $\X$ be the set of all random variables in an atomless probability space $(\Omega,\mathcal F,\p)$. 
  Let $L^\infty$ be the set of essentially bounded random variables in $\X$ and $L^1$ be the set of integrable ones in $\X$.
  Further,
  $L_+^1$ (resp.~$L^\infty_+$) is the set of nonnegative random variables in $L^1$ (resp.~$L^\infty$).  
  Almost surely equal random variables are treated as identical. We write $X\laweq Y$ if $X$ and $Y$ have the same distribution.
  We identify constant random variables with elements in $\R$. 
 Denote by $x\vee y =\max\{x,y\}$
 and 
 $x\wedge y =\min\{x,y\}$ for real values $x$ and $y$.

A risk measure is a functional $\rho:\L\to (-\infty,\infty]$ that is finite on $\R$,   where the domain $\L$ is a convex cone of random variables in $\X$ containing $\R$; examples of $\L$ are $\X$, $L^1$ or $L^\infty$. 
We take the interpretation 
that a random variable in $\L$ represents loss/profit of a financial position (following e.g., \cite{MFE15}). 
A \emph{coherent risk measure} is a risk measure $\rho $ satisfying  the  following four properties (called axioms in the literature).
\begin{itemize} 
    \item[] Monotonicity: $\rho(X)\leq \rho(Y)$ for all $X,Y\in \L$ with $X\leq Y$.
    \item[] Cash invariance: $\rho(X+c)=\rho(X)+c$ for all $X\in  \L$ and $c\in\mathbb{R}$.
    \item[] Convexity: $\rho(\lambda X+(1-\lambda) Y)\leq \lambda \rho(X)+(1-\lambda)  \rho(Y)$ for all $X,Y\in  \L$ and $\lambda \in (0,1)$.
    \item[] Positive homogeneity: $\rho(\lambda X)=\lambda \rho(X)$ for all $X\in  \L$ and $\lambda > 0$.
\end{itemize}
Moreover, a \emph{convex risk measure} is a risk measure that satisfies the first three properties in the above list (\cite{FS02} and \cite{FR02}), and a \emph{monetary risk measure} is a risk measure that satisfies the first two properties in the above list. 
Coherent risk measures automatically satisfy normalization, that is, $\rho(0)=0$.
We refer to \cite{FS16} for interpretations of these properties, which are nowadays standard in the field. 
Two further properties below are also standard and useful in this paper. 
\begin{itemize} 
    \item[] Law invariance: $\rho(X)= \rho(Y)$ for all $X,Y\in \L$ with $X\laweq Y$.
 \item[] Lower semicontinuity (for $L^1$ convergence): $\liminf_{n\to \infty} \rho(X_n)\ge \rho(X)$ if $X_n\to X$ in $L^1$  as $n\to \infty$. 
 \item[] Comonotonic additivity: $\rho(X)+\rho(Y)=\rho(X+Y)$ whenever $X, Y \in \mathcal L$ are comonotonic, i.e., $X$ and $Y$ are both increasing functions of some random variable $U$.
\end{itemize}
Lower semicontinuity in the paper always refers to $L^1$-convergence (omitted), but this does not require $\mathcal L=L^1$. 
%For instance, some risk measures below are well-defined on $\X$  and they satisfy lower semicontinuity.}  
Two popular classes of  monetary risk measures  in financial practice   are  VaR
and  ES,
 which are law-invariant, lower semicontinuous, and defined on $\X$. 
VaR  at level $p\in (0,1)$ is defined as
$$
\VaR_p(X)= \inf \{x\in \R: \p(X\le x)\ge p\},~~~X\in \X,
$$ 
and ES at level $p\in (0,1)$ as 
  $$
\ES_p(X)=\frac 1 {1-p} \int_p^1 \VaR_q(X) \d q ,~~~X\in \X.
$$  
Note that $\ES_p(X)$  is  finite if and only if $\E[X_+]<\infty$.   Either VaR or ES can be characterized by a few   axioms; see \cite{KP16} and \cite{WZ21}.

Both VaR and ES belong to the more general class of distortion risk measures, which plays an important role in this paper. 
Let $\mathcal D$ be the set of 
 functions $h:[0,1]\to[0,1]$ that are increasing (in the non-strict sense)  with $h(0)=0$ and $h(1)=1$. 
   For $h\in \mathcal D$, the Choquet integral of a random variable $X$ with respect to $h\circ \p$ is given by
$$\int X \d (h\circ \p) = \int_0^\infty h\big(\p(X>x)\big)\d x + \int_{-\infty}^0 \Big(h\big(\p(X>x)\big)-1\Big)\d x,
$$
which may be undefined or infinite.   If $X$ is nonnegative, then $\mathbb{P}(X > x)=1$ for all $x \in (-\infty, 0]$ and we are left with only the first integral.  
 For $h\in \mathcal D$,  the \emph{distortion risk measure} $\rho_h$ is defined by   
 %  \begin{equation}
 %  \label{eq:choquet-def}
 $$
   \rho_h(X) = \int X \d (h\circ \p),~~~~X\in \L, 
 %  \end{equation}
 $$
   where $\L$ is chosen such that $\int X \d (h\circ \p)$ takes values in $(-\infty,\infty]$, i.e., the second term in the Choquet integral is finite. 
   The function $h$ is called the \emph{distortion functions} of $\rho_h$.
The following properties of distortion risk measures and Choquet integrals are well known (see \cite{WWW20} for a summary) and will be used frequently in the paper.  
\begin{enumerate}[(a)]
\item  The Choquet integral $\int X \d (h\circ \p)$ is  well-defined and finite on $L^\infty$. Therefore, we will always assume $\mathcal L\supseteq L^\infty$ below implicitly. \item  The class of distortion risk measures $\rho_h$ on $L^\infty$ is precisely the class of law-invariant, monotone, and comonotonic-additive mappings $\rho:L^\infty \to \R$ with $\rho(1)=1$.
 \item  If  $h$ is left-continuous, then we can write 
   $$
   \int X\d (h\circ\p) =\int_0^1   \VaR_{1-q}(X) \d h(q), \mbox{~~~for each $X$ such that $  \int X\d (h\circ\p)$ is well-defined}. 
   $$  
\item  If $h$ is concave, then we can choose $L^1\subseteq \mathcal L$   and $\rho_h$ is lower semicontinuous. 
\item  The risk measure  
   $\rho_h$   is coherent  if and only if $h$ is concave. 
 \item  If $h$ is    the identity on $[0,1]$, then $\rho_h=\E$; hence $\rho_h \geq \E$ if $h $ is concave. \label{f}
 \item  If $h(t) = (t/(1-p))\wedge 1$ for $p\in (0,1)$, then $\rho_h=\ES_p$, which is a coherent risk measure.  
   \end{enumerate}
   For all coherent distortion risk measures   in this paper, their 
  domain will be chosen  as $L^1$, although we initially formulated $\ES_p$ on the larger set $\X$.

\section{Folding score} 
  \label{sec:general}   
  
%\subsection{Folding score of distortion risk measures}  

In this section, we introduce a technical tool and provide some  results on distortion risk measures, which will be used in the proofs of our main results in Section \ref{sec:UI} on uniform integrability. 

Let $\rho$ be a risk measure on  $\L$. 
A new quantity $s_\rho$ will be useful in our later analysis, which is defined by 
   $$s_{\rho} = 
   \sup_{X\in \L} 
   \frac{\rho(|X|)}{    |\rho(X)|\vee |\rho(-X)|}= \sup_{X\in \L} s_\rho(X) ,\mbox{~~~where } s_\rho(X)=  \frac{\rho(|X|)}{   | \rho(X)|\vee |\rho(-X)|},$$
   where we set $\infty/\infty=1$, $0/0=1$, and $1/0=\infty$. 
The quantity $s_\rho$  will be called 
   the \emph{folding score} of $\rho$.
   To explain the name, recall that the distribution of $|X|$ is the folded distribution of $X$. 
Intuitively, $s_\rho$ measures how large $\rho(|X|)$, the risk measure $\rho$ applied to the ``folded" random variable $X$, can be relative to $|\rho(X)|$ and $|\rho(-X)|$. 
Clearly, $s_\rho\ge 1$ for  common risk measures, e.g., those satisfying $\rho(x)=x$ for $x\in \R$. 
We are particularly interested in whether $s_\rho$ is finite.    
   If $\rho=\E^Q$ for a probability $Q$ (absolutely continuous with respect to $\p$), then $s_\rho$ is infinite, because the denominator can be $0$ while the numerator  is positive;  for instance, this happens for any non-degenerate  random variable symmetric about $0$.
Thus, for a linear functional $\rho=\E^Q$, we have $s_\rho=\infty$.

If $\rho$ is a normalized convex risk measure, then $\rho(X)+\rho(-X)\ge 2\rho(0)=0$ for all $X\in \L$. Therefore, we can remove the absolute values in the denominator in the definition of $s_\rho$ and $s_\rho(X)$, that is,
   $$s_{\rho} = 
   \sup_{X\in \L} 
   \frac{\rho(|X|)}{     \rho(X) \vee  \rho(-X) }\mbox{~~~and~~~} s_\rho(X)=  \frac{\rho(|X|)}{    \rho(X)\vee \rho(-X)}.$$
   We will conveniently work with this formulation since most of our results are obtained for coherent risk measures. 
   
%    The intuition behind this name is 
% that, suppose that $\rho(X)=-\rho(-X)$ for some $X\in \L$ such that $|X|>0$. For common risk measures, this implies $\rho(|X|)>0$.
% In this case, 
   
% Consider a risk measure that is positively homogeneous, that is, $\rho(\lambda X)=\lambda 
%  \rho(X)$ for $X\in \L$ and $\lambda>0$.
% Then, $s_\rho$ is very large if there exists  $X\in \L$ such that $\rho(|X|)=1$ and 
% $\rho(X)$ and $\rho(-X)$ are small.
% In particular, if $s_\rho=\infty$, then 
% $\rho(X)$ and $\rho(-X)$ need to be arbitrarily close to $0$. \com{This is insufficient reasoning for the terminology.}
   
%    If $\rho=\E^Q$ for any probability $Q$, then $s_\rho$ is infinite, because the denominator can be $0$ while the numerator  is positive;  for instance, this happens for any  random variable symmetric about $0$.
% Therefore, for a linear functional $\rho=\E^Q$ (which is generally representative for linear risk measures), $s_\rho=\infty$.
% For 

% On the other hand, if $\rho$ is not linear, then we expect to see $s_\rho<\infty$. 

 Our main result in this section is an upper bound on $s_\rho$ for a coherent distortion risk measure $\rho$. In particular, we would like to understand whether $\rho=\E$ is the only case where $s_\rho=\infty$ among distortion risk measures.  This turns out to be true.

   \begin{theorem}\label{th:main}
   Suppose that $\rho$ is a coherent distortion risk measure on $L^1$ with distortion function 
 $h\in \mathcal D$. If $h$ is not the identity, then  
\begin{align}\label{eq:thmain}
   1\le s_{\rho} \le   \frac{h(1/2)+1/2}{h(1/2)-1/2}  < \infty.
   \end{align}
   As a consequence, 
   $s_{\rho}=\infty$ if and only if $\rho$ is the mean. 
   \end{theorem}
   Before proving the theorem, we first present a simple algebraic lemma.
   \begin{lemma}\label{lem:algebraic}
   For $a,b\in [0,1]$, we have 
   $$
   \max_{x,y>0} \frac{x+y}{  (x-ay)\vee( y-bx)} = \frac{2+a+b}{1-ab}.
   $$ 
   \end{lemma} 
   \begin{proof}
With the substitution $z=y/x$,
    $$ \max_{x,y>0} \frac{x+y}{  (x-ay)\vee( y-bx)} 
     =   \max_{z>0} \frac{1+z}{  (1-az)\vee( z-b)}. $$ 
It suffices to maximize $f(z)$ defined by $f(z)= ({1+z} ) /( {  (1-az)\vee( z-b)})$ for $z>0$.      
   If $ab=1$ (i.e., $a=b=1$), then by taking $z=1$ we have $f(z)=\infty$. 
   Next, suppose $ab<1$. 
     Let $z_0=(1+b)/(1+a)$. Note that $ab<1$ implies $z_0>b$, and we have
 $$    f(z_0)= \frac{2+a+b}{1-ab}.$$
 If $z\le z_0$  then $f(z) = (1+z)/(1-az)$, which is increasing in $z$, yielding $f(z)\le f(z_0)$.
 If $z\ge z_0$, then $f(z) = (1+z)/( z-b) = (1-b)/(z-b) +1$, which is decreasing in $z$, yielding $f(z)\le f(z_0)$.
 Therefore, $z=z_0$ maximizes $f(z)$. 
   \end{proof}
   
   \begin{proof}[Proof of Theorem \ref{th:main}]
 Note that if $h$ is concave and not the identity, then $h(q)>q$ for all $q\in (0,1)$, and hence the third inequality in \eqref{eq:thmain} is automatic. 
 The first inequality $1\le s_\rho$ is also automatic. 
 Below, we show the remaining inequality in \eqref{eq:thmain}, that is, $$
\frac{\rho_h(|X|)}{\rho_h(X)\vee \rho_h(-X)} \le   \frac{h(1/2)+1/2}{h(1/2)-1/2}  \mbox{~~~for all $X\in L^1$.}
 $$
 
   Take $X\in L^1$. If $X\ge 0$ or $X\le 0$, then 
      $
    {\rho_h(|X|)} = { \rho_h(X)\vee \rho_h(-X)},
   $
   and the desired inequality   holds. For a general $X$, define the following quantities, 
\begin{align*} 
A &= \int_0^\infty h(\p(X>x))\d x,\\ 
B &= \int_0^\infty h(\p(-X>x))\d x = \int_{-\infty}^0 h(\p(X<x))\d x,\\ 
C &= \int_{-\infty}^0  ( 1-h(\p(X>x)) ) \d x, \\ 
D &= \int_{-\infty}^0 ( 1-h(\p(-X>x)) ) \d x = \int_0^\infty  ( 1-h(\p(X<x))) \d x  .
\end{align*}  
Note that $A, B, C, D$ are all nonnegative, and $C$ and $D$ are bounded from above.
Clearly,
$$
\rho_h(X) = A-C~~\mbox{and}~~\rho_h(-X) = B-D.$$
Moreover,  concavity of $h$ implies $h(s+t)=h(s+t)+h(0)\le h(s)+h(t)$ for $s,t\in [0,1]$ and $s+t\le 1$.
Since $\p(|X|>x) = \p(X>x) + \p(X<-x)$ for $x\ge 0$, we have
\begin{align*}
  \rho_h(|X|) &= \int_0^\infty h\big(\p(X>x) +\p(X<-x) \big)\d x \\
  &\leq \int_0^\infty h\big(\p(X>x)\big ) \d x +  \int_0^\infty  h\big(\p(X<-x)\big)\d x = A + B.
\end{align*}
If one of $A$ and $B$ is infinite, then $\rho_h(|X|)\ge A\vee B=\infty=\rho_h(X)\vee \rho_h(-X)$, and  the desired inequality \eqref{eq:thmain} holds. Below we assume that $A$ and $B$ are finite.

Let $g\in \mathcal D$ be defined by $g(t)=1-h(1-t)$. Clearly, $g$ is convex. 
Noting that $t\mapsto g(t)/t$ is increasing and $t\mapsto h(t)/t$ is decreasing, 
we have that $g/h$ is increasing on $(0,1]$.
Therefore, using the fact that $\p(X=x)>0$ for at most countably many $x\in \R$, we have 
\begin{align*}
\frac{D}{A} &= \frac{\int_0^\infty (1-h(\p(X<x)))\d x}{\int_0^\infty h(\p(X>x))\d x}\\
&= \frac{\int_0^\infty g(\p(X >  x))\d x}{\int_0^\infty h(\p(X>x))\d x}
\le \sup_{x\ge 0}  \frac{ g(\p(X >  x))}{  h(\p(X>x)) } .
\end{align*}
Let $z = \p(X<0)$, we have 
\begin{align*}
  \frac{ g(\p(X>x))}{ h(\p(X>x))} \leq \frac{ g( \p(X> 0))}{ h(\p(X>0))} = \frac{g(1-z)}{h(1-z)}=: b,
\end{align*}
and therefore $ D/A\le b$.
Analogous arguments also yield
\begin{align*}
  \frac{C}{B}   \le \frac{ \int _{-\infty}^0 g(\p(X<x)) \d x    }{ \int _{-\infty}^0 h(\p(X<x)) \d x    } \leq \frac{g(z)}{h(z)} =: a.
\end{align*} 

Note that  both $a$ and $b$ take values in $[0,1)$ since $g$ is  convex and not the identity.
The above inequalities yield
 $$
   \frac{\rho_h(|X|)}{ \rho_h(X)\vee \rho_h(-X)} \le \frac{A+B}{(A-C)\vee (B-D)} \le \frac{A+B}{(A-aB)\vee (B-bA)}. 
   $$
   Using Lemma \ref{lem:algebraic}, we get 
    $$
   s_{\rho_h}(X) \le \frac{2+a+b}{1-ab}.
   $$
   Note that $a$ increases in $z$ and $b$ decreases in $z$. 
  Denote by $c = g(1/2)/h(1/2)$. 
  Since either $z\le 1/2$ or $z > 1/2$, we have either $a\le c$ or $b\le c$.
  This gives $a+b \le 1+ c$ and $ab\le  c$.
  Therefore, 
     $$
   s_{\rho_h}(X) \le \frac{3+c}{1-c} = \frac{3-2g(1/2)}{1-2g(1/2)}= \frac{h(1/2)+1/2}{h(1/2)-1/2}.
   $$
The desired inequality  \eqref{eq:thmain} follows.    \end{proof}

Note that in Theorem \ref{th:main}, concavity of $h$ does not exclude the possibility that $h$ is discontinuous at $0$.
The following proposition states a simple consequence of $s_\rho<\infty$. Our later results will mainly use this proposition.

 \begin{proposition} \label{th:cor1}
 Let $\mathcal S$ be a set of random variables 
 and $\rho$ be a normalized convex risk measure on $\mathcal L $  satisfying  $s_\rho<\infty$ such that $\{X,-X ,|X|\}\subseteq \mathcal L$ for $X\in \mathcal S$.  The following are equivalent.
 \begin{enumerate}[(i)] 
   \item 
   $\rho(X)$ and $\rho(-X)$ are both bounded for $X\in \mathcal S$;
   \item$\rho(|X|)$ is bounded for $X\in \mathcal S$.
 \end{enumerate}
 In particular, this equivalence holds for $\rho=\rho_h$, where
 $h\in \mathcal D$ is concave and not the identity.
 \end{proposition}
 \begin{proof} 
\underline{(i)$\Rightarrow$(ii)}. 
This follows from the definition of $s_\rho$: Note that  
$$\infty>s_\rho \ge \sup_{X\in \mathcal S} \frac{\rho(|X|)}{\rho(X)\vee \rho(-X)}, $$
and hence boundedness of $\rho(X)\vee \rho(-X)$ implies boundedness of $\rho(|X|)$.

\underline{(ii)$\Rightarrow$(i)}. Since the risk measure $\rho$ is monotone, we have $\rho(|X|)\ge \rho(X)$ and $\rho(|X|)\ge \rho(-X)$; hence, both  $\rho(X)$ and $\rho(-X)$ are bounded from above. 
Convexity of $\rho$ gives $\rho(X)+\rho(-X)\ge 2\rho(0)=0$. This guarantees that  $\rho(X)$ and $\rho(-X)$ are bounded from below, noting that  each of  $\rho(X)$ and $\rho(-X)$ is bounded from above.  

The last statement for $\rho=\rho_h$ follows from Theorem \ref{th:main}.
 \end{proof}

 The condition 
 $\{X,-X ,|X|\}\subseteq \mathcal L$
 for $X\in \mathcal S$ is needed for $s_\rho(X)$ to be well-defined for $X\in \mathcal S$.
In most results later, we often encounter the situation $\mathcal S\subseteq L^1 \subseteq \mathcal L$, so that this condition is satisfied automatically. 

 Sharpness of the bound in Theorem \ref{th:main} 
 and its extensions to other risk measures are studied in Appendix \ref{app:folding}. 
 In Section \ref{sec:UI}, we will use a  pair of conditions involving $|X|$ and $\pm X$ similar to Proposition \ref{th:cor1} to characterize uniform integrability.

 \begin{remark}
Let $\rho$ be a law-invariant coherent risk measure on $ L^\infty$. 
If 
\begin{align}
    \label{eq:suffcond}
    \mbox{there exists a non-degenerate $X$ such that $\rho(X)=-\rho(-X)$},
\end{align}
then for  $Y=X-\rho(X)$, we have  
$\rho(Y)=0=\rho(-Y)$, but 
$\rho(|Y|)\ge \E[|Y|]>0$.
Therefore, $s_\rho=\infty$
in this case. 
Note that \eqref{eq:suffcond} is sufficient to force $\rho$ to be the mean; see \cite{CMM04}, \cite{BKMS21} and \cite{LM22} for  related results and generalizations. 
Our Theorem \ref{th:main} can be seen as a strengthening of this result for distortion risk measures, because $s_\rho=\infty$ is a weaker condition than \eqref{eq:suffcond}.
 \end{remark}

\section{Uniform integrability}\label{sec:UI}

\subsection{Characterizing uniform integrability using ES}

  A set $\mathcal S$ of random variables 
 is \emph{uniformly integrable} if 
 $$
\sup_{X\in \mathcal S} \E[|X|\id_{\{|X|>K\}}]\to0 \mbox{~as $K\to\infty$}.
 $$
In our case, the random variables in $\mathcal S$ are defined on an atomless probability space $(\Omega,\mathcal F,\p)$
, and clearly any uniformly integrable set is a subset of $L^1$.  Under this setting, uniform integrability of $\mathcal S$ has the following characterization (see \citet[Theorem T19] {M66}):  
\begin{align}\label{eq:Meyer}
\begin{aligned}
\mbox{For every $\epsilon > 0$,}&\mbox{ there exists $\delta > 0$ such that 
}  \\&\int_A |X| \d \p < \epsilon \mbox{~for all $X\in \mathcal{S} $ and $ A \in \mathcal{F}$ with $\p (A) \leq \delta$}.
\end{aligned}
\end{align}
Moreover,  we can safely replace $\p(A)\le \delta$ in \eqref{eq:Meyer} by  $\p (A) = \delta$.

We have the following characterization theorem of uniform integrability.
The equivalence of (i) and (ii) appeared in 
\citet[Lemma 7.3]{wmh20}; here we supply a shorter proof. Statement (ii) below is easier to work with in technical analysis, whereas statement (iii) has clearer meaning for financial applications.  
For this result, we treat the domain of $\ES_p$ for $p\in (0,1)$ as $\X$, thus allowing for infinite values of $\ES_p(X)$.
 
 \begin{theorem}\label{th:1}
 Let $\mathcal S$ be a set of random variables. 
 The following statements are equivalent.
 \begin{enumerate}
 [(i)]
 \item  $\mathcal S$ is uniformly integrable;
 \item 
 $  \sup_{X \in \mathcal S}(1-p)\ES_p(|X|)$ tends to $0$
   as $p\uparrow 1$; 
 \item both $\sup_{X \in \mathcal S}(1-p)\ES_p(X)  $ and $ \sup_{X \in \mathcal S}(1-p)\ES_p(-X) 
 $ tend to $0$
   as $p\uparrow 1$.
 \end{enumerate}
 \end{theorem}

 \begin{proof} 
 \underline{(i)$\Leftrightarrow$(ii)}. 
 We will use the following well-known result on the representation of  ES (e.g., Eq.~(3.1) of \cite{EW15}):  For any random variable $X$ and $p\in (0,1)$,
\begin{equation}\label{eq:lem:2}
   \ES_p(X) =\sup \{ \E[X|A]: A \in \mathcal F,~ \p(A) = 1-p\}. 
\end{equation}
Given $\epsilon>0$,  by letting $\delta=1-p$ in (ii), using \eqref{eq:lem:2}, we have  
\begin{align*}
    \sup_{X\in \mathcal{S}} (1-p)\ES_p(|X|) &=  \sup_{X\in \mathcal{S}} \sup_{A\in \mathcal{F}_\delta} \delta \E [X|A]= \sup_{X\in \mathcal{S}} \sup_{A\in \mathcal{F}_\delta} \int_A |X| \d \p,
\end{align*}
where $\mathcal{F}_\delta = \{A \in \mathcal F : \p (A) = \delta\}$. Since uniform integrability (i) is equivalent to   \eqref{eq:Meyer}, we know that it is also to equivalent to (ii).

 \underline{(ii)$\Leftrightarrow$(iii)}.  For $p > 1/2$ the ratio of $\ES_p(|X|) $ to $\ES_p(X) \vee \ES_p(-X)$ is bounded below by $1$ and above by a constant from Theorem \ref{th:main} (Example \ref{ex:1} shows this range is contained in $[1,3]$). As $p \uparrow 1$, this uniform bound applies and hence (ii) and (iii) are equivalent.
 \end{proof}

\subsection{Uniform integrability and distortion risk measures} 
Using the equivalence of (i) and (ii) from Theorem \ref{th:1}, we have the following characterization of uniform integrability in terms of distortion risk measures.   
 In what follows, let $\mathcal D_{c}$ 
be the set of concave functions $h\in \mathcal D$ with $h(t)/t\to \infty$ as $t\downarrow 0$.
Recall that the domain of all coherent distortion risk measures $\rho_h$ is chosen as $L^1$.
%Let $h\in \mathcal D$ be concave. 
%As mentioned in Section \ref{sec:2}, $\rho_h$ is well-defined on $L^1$, and
%$\rho_h\ge \E.$
%Hence, $\rho_h(|X|)=\infty$ if  $X$ is outside $L^1$ (but in the domain $\mathcal L$ of $\rho_h$). 
%In all results below, we can without loss of generality take the domain to be $L^1$ for all $\rho_h$.

% , because all three conditions fail if some $X\in \mathcal S$ is outside $L^1$ but in $\mathcal L$.

 \begin{theorem}\label{th:3} 
   For  any $\mathcal{S}\subseteq L^1$, 
   the following are equivalent. 
   \begin{enumerate}[(i)]
       \item  $\mathcal S$ is uniformly integrable;
 \item  there exists  $h\in \mathcal D_c$  such that $   \rho_h(|X|)   $ is bounded for $X\in \mathcal S$;
  \item   there exists  $h\in \mathcal D_c$  such that $  \rho_h(X)     $ and $   \rho_h(-X)    $ are   bounded for $X\in \mathcal S$.  
   \end{enumerate} 
   \end{theorem}
 \begin{proof}
\underline{(ii)$\Rightarrow$(i)}. We proceed by contrapositive. Suppose that $\mathcal{S}$ is not uniformly integrable. Then there exists $m > 0$ and for each $n\in \N$, there exist $ p_n$ and $X_n$  such that  $(1-p_n)\ES_{p_n}(\vert X_n \vert) > m$. 
 
 We show that $\rho_h(\vert X \vert)$ is unbounded for $X\in \mathcal{S}$. That is, given $M> 0$, we construct some $X$ such that $\rho_h(\vert X \vert) > M.$  Choose $p_N$ in the sequence $p_n$ such that $1-p < M/m$. Let $p = p_N$ and $X = X_N$. 
 Denote by \begin{align}\label{eq:hp} h_p(t) = \big(t/(1-p)\big) \wedge 1,\end{align} that is,  the distortion function corresponding to $\ES_p$. We have
 \begin{align*}
 \rho_h(|X|) &= \int_0^\infty h\big(\p(|X|>x)\big)\d x \\
 &\geq \int_0^\infty h\big(\p(|X|>x)\wedge (1-p)\big)\d x\\
 &\geq \frac M m \int_0^\infty \p(|X|>x)\wedge (1-p)\d x\\
 &\geq\frac M m \int_0^\infty (1-p)h_p\big(\p(|X|>x)\big)\d x \geq \frac M m(1-p)\ES_p(|X|)\geq M.
\end{align*}
 Hence, $\{\rho_h(|X|): X \in \mathcal{S}\}$  is not bounded.

\underline{(i)$\Rightarrow$(ii)}. 
Using Theorem \ref{th:1}, (i) implies 
$\sup_{X\in \mathcal{S}} (1-p)\ES_p(\vert X \vert) \downarrow 0$ as $p \uparrow 1$. We can choose a sequence $(p_n)_{n\in\N} \subseteq (0,1)$ such that $1-p_n < 2^{-n}$ and $\sup_{X\in \mathcal{S}} (1-p_n)\ES_{p_n}(\vert X \vert) < 2^{-n}$ for each $n$.
Define for $t\in [0,1]$,  with $h_p$ defined in \eqref{eq:hp}, 
$$g(t)=\sum_{n=1}^\infty (1-p_n)h_{p_n}(t) = \sum_{n=1}^\infty t \wedge (1-p_n).$$
It is straightforward to verify that 
$g$ is finite as $g(t) \leq g(1) =  \sum_{n=1}^\infty (1-p_n) \leq \sum_{n=1}^\infty 2^{-n} = 1$ for $t\in [0,1]$. It is concave, because it is  the sum of concave functions. 
Moreover, $g(0) = 0$ and $g(1)>0$. 
We further define $$h(t) = \frac{g(t)}{g(1)} = \frac{\sum_{n=1}^\infty t \wedge (1-p_n)}{\sum_{n=1}^\infty 1 \wedge (1-p_n)},$$
which is simply the normalized version of $g$ to ensure $h(1) = 1$. The   properties for $g$ imply  that 
$h$ is a distortion function that is concave.  
Moreover, we can verify that for any $N\in \N$, $$\lim_{t\downarrow 0} \frac{g(t)}{t} = \lim_{t\downarrow 0}\sum_{n=1}^\infty 1 \wedge \frac{1-p_n}{t}
\ge \lim_{t\downarrow 0}\sum_{n=1}^N 1 \wedge \frac{1-p_n}{t}
=N.
$$  
Therefore, $\lim_{t\downarrow 0} g(t)/t=\infty.$ 
Finally, 
\begin{align*}
 \rho_h(|X|) &= \frac{1}{g(1)}\int_0^\infty g\big(\p(|X|>x)\big)\d x \\
 &=\frac{1}{g(1)}  \int_0^\infty \sum_{n=1}^\infty \p(|X|>x) \wedge (1-p_n) \d x \\
 &=\frac{1}{g(1)}  \sum_{n=1}^\infty (1-p_n) \ES_{p_n}(\vert X \vert)  \leq  \frac{1}{g(1)} \sum_{n=1}^\infty 2^{-n} =\frac{1}{g(1)} < \infty .
\end{align*}
Therefore, (ii) holds.

\underline{(ii)$\Leftrightarrow$(iii)}. This follows from Proposition \ref{th:cor1}. 
  \end{proof}

The condition $\mathcal S\subseteq L^1$ is imposed only to make sure that $\rho_h$ is properly defined on $\mathcal S$.
Indeed, for $\mathcal S$ that is not contained in $L^1$, each statement in  Theorem \ref{th:3} fails hold, but to make sense of (ii) and (iii), we need to enlarge the domain $L^1$ of $\rho_h$  to include the random variables $|X|,X,-X$ for all $X\in \mathcal S$.

  Theorem \ref{th:3} closely resembles  the classic de la Vall\'ee Poussin
criterion  on uniform integrability: A set $\mathcal S$ is uniformly integrable if and only if there exists an increasing convex function $\phi:\R_+\to \R_+$ with $\phi(x)/x\to \infty$ as $x\to \infty$ 
  such that $\sup_{X\in \mathcal S}\E[\phi(|X|)]<\infty$; see \citet[Theorem T22]{M66}. This highlights a duality between $X\mapsto \rho_h(X)$ for a concave $h$ and $X\mapsto \E[\phi(X)]$ for a convex $\phi$. Indeed, $\rho_h$ is called the dual utility model of \cite{Y87} in decision theory, and $X\mapsto\E[\phi(X)]$ represents an expected utility model in the same context.

\begin{remark} 
For the implication in  Theorem \ref{th:3} (ii)$\Rightarrow$(i), 
\begin{align}
    \label{eq:R1-1}
\mbox{$\rho_h(|X|)$ is bounded for $X\in \mathcal S\subseteq L^1$} ~\Longrightarrow ~\mbox{$\mathcal S$ is uniformly integrable}, 
\end{align}
the risk measure $\rho_h$ satisfies convexity, cash invariance, and positive homoegeneity. 
It is clear that if we replace $\rho_h$ by 
$\rho=f\circ \rho_h$ where $f:\R\to \R$ is an increasing function that is unbounded from above, then \eqref{eq:R1-1} remains true for $\rho$. Such $\rho$ does not necessarily satisfy the properties of $\rho_h$. 
Hence, for the implication (ii)$\Rightarrow$(i) in Theorem \ref{th:3}, the essential property that we need is how  $\rho_h$ integrates  the tail probability, and \eqref{eq:R1-1} can hold for risk measures without convexity, cash invariance, or positive homoegeneity. 
\end{remark}
  
The condition $h\in \mathcal D_c$ in Theorem \ref{th:3} allows for $h$ to have a jump at $0$, which makes $h(t)/t\to \infty$ automatically as $t\downarrow 0$.
Such $h$ is not interesting, because boundedness of $\rho_h(|X|)$ over $X\in \mathcal S$ for such $h$ forces   $\mathcal S$ to be bounded, an obvious case of uniform integrability.  

In part (ii) or (iii) of  Theorem \ref{th:3}, $\rho_h$ 
 serves as a ``testing" risk measure for uniform integrability of $\mathcal S$. In the statement of Theorem \ref{th:3}, such a choice of $\rho_h$   depends on $\mathcal S$. 
 One may wonder whether 
 there is   a single $\rho_h$ that works for all sets $\mathcal S$. 
 To answer this question, we first obtain a result on the finiteness of coherent distortion risk measures, which may be of independent interest. 
For a  risk measure $\rho$ on $L^1$, we say that $\rho$ is \emph{expectation-dominated} if there exists $c>0$ such that 
$\rho \le c \, \E$ on $L_+^1$. 
 For instance, $\ES_p$ is expectation-dominated for each $p\in (0,1)$, since $\ES_p(X)\le (1-p)^{-1}\E[X]$ for $X\in L^1_+$.

 \begin{theorem}\label{th:finite}
Let $\rho$ be a coherent distortion risk measure on $L^1$.
The following are equivalent.
\begin{enumerate}[(i)]
    \item  $\rho$ is expectation-dominated;
    \item $\rho$ is finite on $L^1$;
    \item the distortion function $h$ of $\rho$ satisfies $\lim_{t\downarrow 0} h(t)/t<\infty$.
\end{enumerate}
 \end{theorem}
 \begin{proof}
 \underline{(i)$\Rightarrow$(ii)}.
Since $\rho\le c\E$ on $L^1_+$ for some $c>0$, we have    $  \rho(|X|)<\infty$
for all $X\in L^1$. 
Since $\rho$ is convex,
$\rho(-|X|)+\rho(|X|)\ge 2\rho(0)= 0$, which implies $ \rho(-|X|)>-\infty$.
Monotonicity of $\rho$ yields $-\infty<\rho(-|X|)\le \rho(X) \le \rho(|X|) <\infty$, showing that $\rho$ is finite on $L^1$.

\underline{(ii)$\Rightarrow$(iii)}.  
Note that $\lim_{t\downarrow0}h(t)/t$ exists since $h(t)/t$ is decreasing in $t$, due to concavity of $h$. 
Suppose $h(t)/t=\infty$, and  we will show that $\rho$ cannot be finite on $L^1$.   
Let $U$ be uniformly distributed on $[0,1]$. 
For $t\in (0,1)$, let  $B_t=\id_{\{U<t\}}$ and   
$$X_t=\bigg(\frac{1}{th(t)}\bigg)^{1/2} B_t .$$
Note that $\E[X_t]=(t/h(t))^{1/2}\le 1.$
It is straightforward to compute 
$\rho(X_t)= (h(t)/t)^{1/2}$
for $t\in (0,1)$,
and hence  
$$
\lim_{t\downarrow 0} \rho(X_t) 
=\lim_{t\downarrow 0} \bigg(\frac{h(t)}{t}\bigg)^{1/2}=\infty.
$$
Let $t_n >0$ be such that $\rho (X_{t_n})>2^n.$
Define a random variable  $X$ by 
$$X= \sum_{n=1}^\infty 2^{-n} X_{t_n}.$$
Note that $X$ is in $L^1_+$ since $X\ge 0$ and  $\E[X] = \sum_{n=1}^\infty 2^{-n} \E[X_{t_n}] \le 1.$
Finally,
we note that the random vector $(X_{t_n})_{n \leq k}$ is comonotonic for each $k \in \N$ because each $X_t$ is a decreasing function of $U$. Since any distortion risk measure is comonotonic-additive, monotone, and positively homogeneous,
we have, for each $k$,
$$\rho(X) = \rho\bigg(  \sum _{n=1}^\infty 2^{-n}X_{t_n}\bigg )\ge  \rho\bigg(\sum_{n=1}^k  2^{-n} X_{t_n}  \bigg)
=   \sum_{n=1}^k  2^{-n} \rho(X_{t_n}) >  \sum_{n=1}^k   1 =k.
$$
This shows $\rho(X)=\infty$, and hence $\rho$ is not finite on $L^1$.

 \underline{(iii)$\Rightarrow$(i)}. 
Let  $c=\lim_{t\downarrow 0} h(t)/t<\infty$.
  Since $h(t)/t$ is decreasing in $t$, we have  
  $h(t)\le c t$ for $t\in [0,1]$.
Hence, $\rho_h(X) \le c \E[X]$ for $X\in L^1_+$. 
 \end{proof}

 Theorem \ref{th:finite}  generalizes a special case of Proposition 1 of \cite{WWW20}, which states that if $h$ is concave and $\rho_h$ is finite on $L^1$, then the right derivative $h'$ of $h$ has finite $L^q$-norm on $[0,1]$ with respect to the Lebesgue measure for all $q>0$.
 Theorem \ref{th:finite} further gives  that $h'$ has finite $L^\infty$-norm, which is a stronger condition.

 \begin{remark}
 Theorem \ref{th:finite} (i)$\Leftrightarrow$(iii) can be derived from
 some results in the context of rearrangement-invariant Banach norms. In particular, we find Theorem 29 and Corollary 30 of \cite{FW22} the most relevant. Translating their results on Banach norms into our setting of distortion risk measures, it is shown that for $h\in \mathcal D$ with  $h(0+)=0$,  if $ h'(0)<\infty$ then the induced norm by $\rho_h$  via $X\mapsto \rho_{h}(|X|)$ is equivalent to the $L^1$ norm; if  $h'(0)=\infty$ then the induced norm by $\rho_h$ is not equivalent to the $L^1$ norm. Our proof uses an explicit construction to show the implication (ii)$\Rightarrow$(iii), not covered by \cite{FW22},   and includes the case $h(0+)>0$. %Moreover, the equivalent statement (ii) is new. 
 A similar observation can be made for Proposition \ref{prop:finite2} below, which characterizes the finiteness of law-invariant coherent risk measures on $L^1$. 
 \end{remark}
 
Using Theorem \ref{th:finite}, we can show that there does not exist a single $\rho_h$ that works for all sets $\mathcal S$ in Theorem \ref{th:3}.  
 
\begin{proposition}
There is no coherent distortion risk measure $\rho$
 such that for all sets $\mathcal S \subseteq L^1 $, 
 $$
 \sup_{X\in \mathcal S} \rho(|X|)<\infty~~ \iff ~~\mbox{$\mathcal S$ is uniformly integrable}.
 $$  
\end{proposition}
\begin{proof} 
First, by Theorem \ref{th:finite}, if $\rho$ is finite on $L^1$,  then it is expectation-dominated.
Together with property (\ref{f}) from Section \ref{sec:2}, we have $\E \leq \rho \leq c\E $ on $L_+^1$ for some $c\ge 1$. In this case, $ \sup_{X\in \mathcal S} \rho(|X|)<\infty$ is equivalent to $\sup_{X\in \mathcal S} \E[|X|]<\infty$, which is not sufficient for uniform integrability of $\mathcal S$. 
Next,  by Theorem \ref{th:finite} again, if $\rho$ is not finite on $L^1$, then there exists $X\in L^1$ such that $\rho(|X|)=\infty$,
but $\{X\}$ is uniformly integrable. Hence, uniform integrability of $\mathcal S $ does not guarantee $\sup_{X\in \mathcal S} \rho_h(X) <\infty.$
% Therefore, it suffices to consider $h\in \mathcal D_c$. 
% Theorem \ref{th:finite} 
% guarantees taht $\rho_h$
% For $h\in \mathcal D_c$, we will show that uniform integrability of $\mathcal S$ is not sufficient for $ \sup_{X\in \mathcal S} \rho_h(|X|)<\infty$. 
% For $t\in (0,1)$, let $B_t$ be a Bernoulli random variable with mean $t$, and   
% $$X_t=\left(\frac{1}{th(t)}\right)^{1/2} B_t .$$
% For $\theta>0$, we can compute 
% $$
% \sup_{t\in (0,1)} \E\left[|X_t|\id_{\{|X_t|>1/\theta \}}\right]
% =\sup_{t\in (0,1)} \left (\frac{t}{h(t)  }
% \right)^{1/2}  \id_{\{h(t) t < \theta ^2\}}.
% $$
% Note that since $h(t)\ge t$ for $t\in (0,1)$, we have
% $\sup\{t>0:  h(t) t < \theta^2 \}
% \le \sup\{t>0:  t < \theta \}
% =\theta $. Therefore, 
% $$
% \sup_{t\in (0,1)} \E\left[|X_t|\id_{\{|X_t|>1/\theta \}}\right]
% \le \sup_{t\in (0,\theta)} \left (\frac{t}{h(t)}
% \right)^{1/2} = \left (\frac{
% \theta }{h(\theta)} 
% \right)^{1/2} \to 0 \mbox{~as $\theta \downarrow 0$}. 
% $$
% Hence, the set $\mathcal S=\{X_t: t\in (0,1)\}$ is uniformly integrable. 
% On the other hand, 
% It is straightforward to compute 
% $\rho_h(X_t)= (h(t)/t)^{1/2}$
% for $t\in (0,1)$,
% and hence  
% $$
% \sup_{X\in \mathcal S} \rho_h(X) 
% =\sup_{t>0} \left(\frac{h(t)}{t}\right)^{1/2}=\infty 
% $$ by  using $h\in \mathcal D_c$. 
% Thus, uniform integrability of $\mathcal S $ does not guarantee $\sup_{X\in \mathcal S} \rho_h(X) <\infty.$
\end{proof}

\begin{example}[Integrated ES]\label{ex:ies}
Define the integrated ES (IES) for a random variable $X$ as
\begin{align}\label{eq:alter}
\mathrm{IES}(X) =\int_0^1 \ES_p (X) \d p   = \int_0^1 \int_0^q  \frac{1}{1-p} \VaR_q (X)  \d p\d q = \int_0^1- \log(1-q)\VaR_q (X) \d q.
\end{align}
The distortion function $h$ of $\mathrm{IES}$ is in $\mathcal D_c$ because 
$\lim_{t\downarrow 0}h(t)/t=  
- \lim_{t\downarrow 0}  \log (t)=\infty. $  
 By Theorem \ref{th:3}, 
any set $\mathcal S$  of random variables on which  $ \mathrm{IES}   $ is  bounded is uniformly integrable. 
  Proposition 1 of \cite{WWW20} implies that $\mathrm{IES}  $ is finite on $ L^{1+\epsilon}$ for any $\epsilon >0$, and by Theorem \ref{th:finite}, $\mathrm{IES}  $ is not finite on $ L^{1}$. One specific example of $X\in L^1$ with $\mathrm{IES}(X)=\infty$ is given below.
  Let $U$ be uniformly distributed on $[0,1/2]$ and let   $X= U^{-1} (\log U )^{-2}    $.
  % be specified  by $\VaR_p(X)=  ( -\log(1-p))^{-2} (1-p)^{-1}.$
  We can compute
\begin{align*}
  \E[X]  & = 2 \int_0^{1/2}   ( - \log u )^{-2}  u^{-1} \d u  = 2 \int_{\log 2}^\infty  t ^{-2}    \d t   = 2 (\log 2 )^{-1}<\infty.
   \end{align*}
   By \eqref{eq:alter} and $\VaR_{1-q}(X)=2q ^{-1} (\log q - \log 2)^{-2}$, 
    \begin{align*} 
  \mathrm{IES}(X)  
=  \int_0^1  
-\log u  \frac{2} {u(\log u-\log 2)^2}\d u=\int_0^ {\infty} \frac{2x}{(x+\log2 )^2}\d x
=\infty.
   \end{align*}
% Using the representation of  ES  in  \eqref{eq:lem:2},
%  \begin{align*}
%    \ES_p(X)  & \ge \E[X|U\le (1-p)/2] 
%    \\& = \frac{2}{1-p }   \int_0^{(1-p)/2}( - \log u )^{-2}  u^{-1} \d u  = \frac{2}{1-p } (\log 2 -\log( 1-p ))^{-1}.
%    \end{align*}
%    By definition we have 
%    \begin{align*} 
%   \mathrm{IES}(X)  
% \ge  \int_0^1 \frac{2}{1-p } (\log 2 -\log( 1-p ))^{-1}  \d p = 2 \int_0^\infty (\log 2  + t)^{-1} \d t=\infty.
%    \end{align*}
 We can see that the set   $\{X\}$  is uniformly integrable and $\mathrm{IES}(X)=\infty$.
   \end{example}

\subsection{Uniform integrability and coherent risk measures}

Next, we turn to the more general class of law-invariant coherent risk measures.  
Let $\mathcal R$ be the set of all lower semicontinuous and law-invariant coherent risk measures on $L^1$ that are not finite on $L^1$. For instance,  IES  in Example \ref{ex:ies} is in $\mathcal R$.

\begin{theorem}\label{th:3-ch} 
For any $\mathcal S\subseteq L^1$, the following are equivalent.  
\begin{enumerate}[(i)]
\item $\mathcal S$ is uniformly integrable;
\item  there exists $\rho\in \mathcal R$ such that   $\rho(|X|) $ is bounded for $X\in \mathcal S$;
\item  there exists $\rho\in \mathcal R$ such that   $  \rho(X)     $ and $   \rho(-X)   $ are   bounded for $X\in \mathcal S$;
\item  there exist $\rho,\rho'\in \mathcal R$ such that   $  \rho(X)     $ and $   \rho'(-X)   $ are   bounded from above for $X\in \mathcal S$.
\end{enumerate}  
   \end{theorem}
   \begin{proof}
   Any coherent distortion risk measure $\rho_h$ is lower semicontinuous. 
By Theorem \ref{th:finite}, $\rho_h \in \mathcal R$  if and only if   $h \in \mathcal D_c$.
Therefore, by using Theorem \ref{th:3} and taking $\rho=\rho_h$, we get (i)$\Rightarrow$(ii)$\Rightarrow$(iii).
The direction (iii)$\Rightarrow$(iv) is trivial. It remains to prove (iv)$\Rightarrow$(i).

%Next, we prove (iii)$\Rightarrow$(i).  
 We first show that for 
  any coherent distortion risk measure $\rho_h$ and $c>0$, we have 
  \begin{align}
  \label{eq:th4-1}
\mbox{$\rho_h\le c\E$ holds on $L^1_+$} \iff 
\lim_{t\downarrow 0} h(t)/t \le c .
\end{align}
A weaker version of this equivalence was used in the proof of Theorem \ref{th:finite}.
Suppose  $\lim_{t\downarrow 0} h(t)/t\le c$.
  Since $h(t)/t$ is decreasing in $t$, 
  we know  
  $h(t)\le c t$ for $t\in [0,1]$.
Hence, $\rho_h(X) \le c \E[X]$ for $X\in L^1_+$.
Conversely, if  $\rho_h\le c\E$ on $L^1_+$, then $\rho_h(B_t)/t \le c $, 
where $B_t$ is a Bernoulli random variable with $\E[B_t]=t$. This implies $h(t)/t$ is bounded above by $c$. Since $h(t)/t$ is  decreasing in $t$, it has a finite limit bounded by $c$ as $t\downarrow 0$. 
 
Note that a lower semicontinuous and law-invariant coherent risk measure $\rho$  on $L^1$ admits a Kusuoka representation (\cite{K01}; see \cite{FS12} for its extension to $L^1$), that is, 
% Next, note that any law-invariant coherent risk measure that is finite on $L^1$ admits a Kusuoka representation (\cite{K01}; see Section 7 of \cite{R13}   for the result on $L^1$), that is, 
\begin{align}\label{eq:coh-rep}
\rho= \sup_{h\in\mathcal H_\rho} \rho_h \mbox{~
for a set $\mathcal H_\rho$ of concave distortion functions}.
\end{align} 
Since $\rho$ is not finite on $L^1$, it cannot be expectation-dominated; this follows by the same argument of {(i)$\Rightarrow$(ii)} in the proof of Theorem \ref{th:finite}. That is, for every $c>0$, $\rho \le c\E $ on $L^1_+$ does not hold.
Hence, for each $n$,
there exists $h_n\in \mathcal H_\rho$ such that $\rho_{h_n}\le 2^n \E $ on $L^1_+$ does not hold, and by using \eqref{eq:th4-1}, this implies $\lim_{t\downarrow 0} h_n(t)/t>2^n.$ 
Write $g=\sum_{n=1}^\infty 2^{-n} h_{n}$.
Since $g$ is a convex combination of concave distortion functions, it is also a concave distortion function. 
Moreover, we can compute (via Fubini's theorem)
$$
\lim_{t\downarrow 0} \frac{g(t)}{t}
=  \sum_{n=1}^\infty  2^{-n} \lim_{t\downarrow 0} \frac{h_{n}(t)}{t}  
\ge \sum_{n=1}^\infty 1 =\infty.
$$
Therefore, $g\in \mathcal D_c$.   Moreover, $\rho\ge \rho_g$ since $\rho\ge \rho_{h_{n}}$  for each $n\in \N$. 
Similarly, there exists $f\in \mathcal D_c$ such that $\rho'\ge \rho_f$. 

% Let $\mathcal D_{c}$ 
% be the set of concave functions $h\in \mathcal D$ with $h(t)/t\to \infty$ as $t\downarrow 0$. 

Take $\ell = g\wedge f$.
We can see that $\ell$ is a concave distortion function, and $\ell(t)/t = (g(t)\wedge f(t))/t \to \infty$ as $t\downarrow 0$.
Therefore, $\ell \in \mathcal  D_c$.
Moreover, $\rho_\ell\le \rho_g\wedge \rho_f$ by definition of $\ell$.
%Hence, $\rho_\ell(X)\le \rho_g(X)$
%and $\rho_\ell(-X) \le \rho_f(-X)$
%and these two terms are both bounded from above for $X\in \mathcal S$.   
Hence,   
the boundedness  from above of $\rho(X)$ and $\rho'(-X)$  for $X\in \mathcal S$ implies that $\rho_\ell(X)$ and $\rho_\ell(-X)$ are both bounded above; they are also bounded below because $\rho_\ell(X)+\rho_\ell(-X)\ge 0$ for all $X\in L^1$. 
By  using Theorem \ref{th:3},  we know that $\mathcal S$ is uniformly integrable. 
   \end{proof}

  The implication (iv)$\Rightarrow$(i) will be useful in the application in Section \ref{sec:investment}, where one of the risk measures represents a price functional.

 Let $\rho $ be a lower semicontinuous and law-invariant coherent risk measure on $L^1$. 
 In  the direction (iii)$\Rightarrow$(i) of Theorem \ref{th:3-ch}, it is required that $\rho$ is not finite on $L^1$. It may be natural to ask whether this requirement is also necessary. 
 Clearly, if $\rho$ is expectation-dominated, then $\sup_{X\in \mathcal S}\rho(|X|)<\infty$ for any $\mathcal S\subseteq L^1_+$ with bounded expectation, which is not sufficient for uniform integrability. Therefore, for the implication (iii)$\Rightarrow$(i),   $\rho$ is necessarily  not expectation-dominated. 
It turns out that this is equivalent to requiring $\rho$ to be not finite on $L^1$.
The next result on the finiteness of such risk measures extends Theorem \ref{th:finite} to the larger class of law-invariant coherent risk measures. For other results on spaces that are finite for law-invariant risk measures, see \cite{LS17}.
\begin{proposition}\label{prop:finite2}
A law-invariant coherent risk measures on $L^1$ is finite if and only if 
it is expectation-dominated. 
\end{proposition}
\begin{proof}
The ``if" statement follows from the same argument of {(i)$\Rightarrow$(ii)} in the proof of Theorem \ref{th:finite}.
Below we show the ``only if" statement. 
 Note that any law-invariant coherent risk measure that is finite on $L^1$ admits a Kusuoka representation in the form of \eqref{eq:coh-rep}.  
From the arguments in the proof of Theorem \ref{th:3-ch}, we see that $\rho\ge \rho_g$ for some $g\in \mathcal D_c$.    By Theorem \ref{th:finite}, this implies that $\rho_g$ is not finite on $L^1$.  Hence, we conclude that $\rho$ is also not finite on $L^1$. 
% This shows that if $\rho$ is finite on $L^1$, then $\rho$ is expectation-dominated. 
\end{proof}

 Let $\rho $ be a lower semicontinuous and law-invariant coherent risk measure.
 Proposition \ref{prop:finite2} implies that the set $\mathcal R$ contains precisely such $\rho$ that are not expectation-dominated. Moreover, the non-finiteness of $\rho$ on $L^1$ is necessary and sufficient 
 for the implication (iii)$\Rightarrow$(i).

   \section{Some consequences of the main results}
\label{sec:application}
   % Theorem \ref{th:3} allows us to state a law of large numbers. The Appendix lists several related results. In particular, part (1) of theorem \ref{th:SL} together with theorem \ref{th:3} proves the following theorem. 

   In this section we present two applications, both following directly from the condition of uniform integrability.
   The first corollary concerns a weak law of large numbers. 

\begin{corollary}\label{cor:lln}
Let  
 $(X_n)_{n\in \N}$ be  a sequence  of pairwise independent random variables with zero mean and $Y_n = (X_1 + \dots + X_n)/n$ for each $n\in \N$. If $\sup_{n\in \N} (\rho(X_n)
 \vee\rho'(-X_n)) <\infty$ for some $\rho,\rho' \in\mathcal R$, then $Y_n\to 0$ in probability  as $n\to\infty$.    
\end{corollary}

\begin{proof}
For a sequence of pairwise independent random variables with zero mean, uniform integrability ensures the weak law of large numbers; see \cite{LR87}. 
The statement then follows from Theorem \ref{th:3-ch}.
\end{proof}
 
\begin{remark} \label{rem:2} It is known that, for a sequence of independent random variables with zero mean,
uniform integrablity is not sufficient for the strong law of large numbers; see the example in \cite{LR87}. Since our condition in Corollary \ref{cor:lln} is equivalent to uniform integrability,  we cannot change the statement of convergence in probability to that of almost sure convergence.
A remaining question is to find a stronger condition on $\rho $ than $\rho\in \mathcal R $ such that 
$$
\sup_{n\in \N} (\rho(X_n)
 \vee\rho(-X_n)) <\infty ~~\Longrightarrow ~~\frac 1n \sum_{i=1}^n X_i \to 0 \mbox{~a.s.}
$$
for an independent (or pairwise independent) sequence $(X_n)_{n\in \N}$ with zero mean. 
\end{remark}

%\subsection{Convergence of ES for a sequence of random variables}
The second corollary concerns the convergence of ES under convergence in distribution.

\begin{corollary}
Suppose that  a sequence 
 $(X_n)_{n\in \N}$  converges to $X\in L^1$ in distribution and satisfies $\sup_{n\in \N} (\rho(X_n)
 \vee\rho'(-X_n)) <\infty$ for some $\rho,\rho' \in\mathcal R$. Then $\ES_p(X_n)\to \ES_p(X)$  for all $p\in (0,1)$.  
\end{corollary}
\begin{proof}
%Let $F_n$ be the distribution of $X_n$ for each $n\in \N$, and $F$ be the distribution of $X$. 
Let $U$ be a uniform random variable on $[0,1]$. 
For $n\in \N$, 
denote by $Q_n(t)=\VaR_t(X_n)$ 
and $Q(t)= \VaR_t(X)$ for $t\in (0,1)$. Note that 
$Q_n(U) \laweq X_n$
and $Q(U)\laweq X$. 
Using Theorem \ref{th:3-ch}, we know that $\{X_n:n\in\N\}$ is uniformly integrable, so 
$\{Q_n(U):n\in\N\}$ is also uniformly integrable.
 The convergence $X_n\to X$ in distribution as $n\to \infty$   implies that $Q_n(U)\to Q(U)$ in probability.
This and uniform integrability of $\{Q_n(U):n\in\N\}$ guarantee that $Q_n(U)\to Q(U)$ in $L^1$. 
Hence, 
$\int_0^1 |Q_n(u)-Q(u)| \d u \to 0, $
which implies 
$\int_p^1 |Q_n(u)-Q(u)| \d u \to 0.$
Therefore, by law invariance of ES, we have 
$\ES_p(X_n)=\ES_p(Q_n(U))\to \ES_p(Q(U))=\ES_p(X).$ 
    % This is a direct consequence of the fact that $\ES_p$ is continuous in Wasserstein $L^1$ convergence, and the fact that convergence in distribution under uniform integrability is sufficient for Wasserstein $L^1$ convergence. 
\end{proof}

Next, we present a result that will be useful in the application in Section \ref{sec:investment}.  
Define the $1$-Wasserstein distance  between two distributions $F$ on a space  $\mathfrak X$ and $G$ on a space  $\mathfrak Y$ by
$$
W^1(F,G) = \inf_{\pi \in \Pi(F,G)} \int_{\mathfrak X\times \mathfrak Y} |x-y|  \pi (\d x,\d y),
$$
where $\Pi(F,G)$ is the set of distributuons on $\mathfrak X\times \mathfrak Y$ with marginal distributions $F$ and $G$. 
The $1$-Wasserstein distance belongs to the  class of $p$-Wasserstein distances 
commonly used in optimal transport theory (e.g., \citet[Chapter 2]{R13}). 
When $F$ and $G$ are distributions on $\R$, we identify them with the corresponding cumulative distribution functions. 
In this case, 
$W^1(F,G)$ has an explicit formula 
$$W^1(F,G) =\int_0^1 |F^{-1}(t) -  G^{-1}(t) |\d t,$$
where $F^{-1}$ is the quantile function of $F$, that is, 
$F^{-1}(t)=\inf \{x\in \R: F(x)\ge t\}$ 
for $t\in (0,1)$. 
We will denote by $w^1$ the corresponding pseudometric on $L^1\times L^1$, that is, $$w^1(X,Y) 
=\int_0^1 |\VaR_t(X) - \VaR_t(Y) |\d t =W^1(F,G),
$$ 
where $F$ and $G$ are the distributions of $X$ and $Y$, respectively. Note that $w^1(X,Y)\le \E[|X-Y|]$ for any $X,Y$, and equality holds if $X$ and $Y$ are comonotonic. Using this fact and \citet[Theorem 4.6.3]{D19}, one can check that,  as $n\to\infty$,
\begin{equation}
\label{eq:R1-2}
\begin{aligned}
 w^1(X_n,X)\to 0 &\iff X^*_n\to X \mbox{ in $L^1$} 
 \\ & \iff X_n
\to X \mbox{ in distribution and $(X_n)_{n\in \N}$  is uniformly integrable},
\end{aligned}
\end{equation}
where $X_n^*$ is identically distributed as $X_n$ 
and $X,X_n^*$ are comonotonic  for each $n\in \N$; such $X_n^*$ can be constructed on the same probability space as $X$ as long as the space is atomless. 
Note that in \eqref{eq:R1-2}, 
$\E[|X_n^*-X|]=w^1(X_n^*,X) \le \E[|X_n-X|]$ and thus the distribution of $(X^*_n, X)$ is a minimizer in the definition of the $1$-Wasserstein distance.
% Note that the joint distribution of $(X_n,X)$ that attains the infimum in the definition of $w^1(X_n, X)$  is that of $(X^*_n, X)$, because $w^1(X_n^*,X)=\E[|X_n^*-X|]$.

\begin{corollary}
\label{coro:R1-1}
Suppose that a sequence 
 $(X_n)_{n\in \N}$   satisfies $\sup_{n\in \N} (\rho(X_n)
 \vee\rho'(-X_n)) <\infty$ for some $\rho,\rho' \in\mathcal R$. Then 
 there exists a subsequence $(X_{n_k})_{k\in \N}$ 
 that converges  
  in $w^1$.  
 \end{corollary}
\begin{proof}
First, Theorem \ref{th:3-ch} guarantees that that $\{X_n:n\in\N\}$ is uniformly integrable.
Note that   uniform integrability implies $ x \sup_{n\in \N}\p(|X_n|>x) \to 0$ as $x\to \infty$, which in term implies $ \sup_{n\in \N}\p(|X_n|>x) \to 0$ as $x\to\infty$. Therefore, the sequence $\{F_n:n\in\N\}$ of distributions of $\{X_n:n\in\N\}$
is tight. 
Hence, $\{F_n: n\in \N\}$ has a convergent subsequence to some distribution $F$; see e.g., Theorem  3.10.3 of \cite{D19}.
Let $U$ be a uniformly distributed random variable on $[0,1]$
and $X=F^{-1}(U)$.
Clearly,
$
F^{-1}_n(U) \to F^{-1}(U) 
$ almost surely (this is implied also by the Skorokhod representation theorem; see Theorem 6.7 of \cite{B99}), and 
the set $\{F^{-1}_n(U):n\in \N\}$
 is uniformly integerable since $\{X_n:n\in\N\}$ is uniformly integrable.
 Therefore, $F^{-1}_n(U) \to X$ in $L^1$ by Theorem 4.6.3 of \cite{D19}. 
 This means $X_n\to X$ in $w^1$ because $w^1(X_n,X)= W^1(F_n,F) =\E[|F^{-1}_n(U)-X|]$.     
\end{proof}

\section{An application in investment optimization}
\label{sec:investment}

We consider an investment problem where a decision maker tries to maximize an expected utility function $u$ of two variables subject to a risk constraint.
We assume that the risk of random losses is assessed by 
  a lower semicontinuous and law-invariant coherent risk measure $\rho$ on $L^1$ that is not finite on $L^1$, i.e., $\rho \in \mathcal R$. 
As an example, $\rho$ may be a distortion risk measure $\rho_h$
with a concave distortion function satisfying $\lim_{t\downarrow 0} h(t)/t\to \infty$, such as  IES in Example \ref{ex:ies}. 
Moreover, we use another functional 
$P\in \mathcal R$ to represent the ask price of a financial asset; that is, $P(-X)$ is the price to purchase a random  loss $X$ with payoff  $-X$. 
Using a coherent risk measure for  the ask price of financial assets is a common setup in \emph{conic finance}; see \cite{MC10} for an introduction. Coherent distortion risk measures in $\mathcal R$, such as \begin{align}
P(X)=\rho_h (X) =\alpha \int_0^1 (1-t)^{\alpha-1}   \VaR_t(X) \d t, \mbox{~~with $h(t)=t^{\alpha}$},
\end{align} 
are used as ask prices in \cite{MC10}. 
To  connect to our main results, we assume that both the price functional and the risk measure are law invariant, and thus we do not consider the pricing kernel.  Note that $P \ge \E $ for
$P\in \mathcal R$, and thus the ask price is more expensive than the mean payoff, a sensible assumption.

%, or $h(t)=t^{\alpha}$ for some $\alpha \in (0,1)$, that is, 
%\begin{align}
 %   \label{eq:PH} 
%\%rho(X)=\rho_h (X) =\alpha \int_0^1 (1-t)^{\alpha-1}   \VaR_t(X) \d t.
%\end{align} 

Let $\mathcal G$ and $\mathcal G^\perp$ 
be two  independent sub-$\sigma$-fields of $\mathcal F$
such that $(\Omega,\mathcal G,\p)$ and $(\Omega,\mathcal G^\perp,\p)$ are atomless.
A stylized investment optimization problem can be formulated 
as  \begin{equation}
    \label{eq:R1-opt}
    \begin{aligned} 
\mbox{maxmize}~~~ & \E[u(-X,Y)]\\
\mbox{subject to}~~~ & X\in L^1(\mathcal G); ~\rho(X)\le r_0;~ P(-X) \le x_0,    
\end{aligned}
\end{equation}
where  
$Y\in L^1(\mathcal G^{\perp})$, 
 $f:\R^2\to\R$ is measurable,
 and $r_0,x_0\in \R$.  
In the model \eqref{eq:R1-opt}, 
\begin{enumerate}[(a)]
    \item 
$u$ represents 
 a bivariate utility function that the decision maker aims to maximize; 
 \item $Y $ represents some independent background risk outside the control of the decision maker (e.g., risks outside the financial market, such as labour income risk, or insurance risk);   
 \item 
  the decision variable $X$ is the risky position (random loss) that the decision maker takes in the financial market;
  \item $r_0$ is the   risk budget for the loss $X$ evaluated by the risk measure $\rho$;  
\item $x_0$ is the   budget of the decision maker priced by $P$;  
 \item $\mathcal G$ is the $\sigma$-field of a complete financial market, so that every payoff $X\in L^1(\mathcal G)$ is attainable. 
\end{enumerate} 
See \cite{BS01} for a similar utility optimization problem in a complete market with VaR and budget constraints. 
Independent background risk is common in  decision problems; see \cite{MPST24} for recent advances. 
In case $Y$ is continuously distributed, one may   conveniently set $\mathcal G^\perp=\sigma(Y)$.

We assume that $u$ is   Lipschiz continuous; that is, there exists $c>0$ such that 
$|u(x,y)-u(x',y')|\le c(|x-x'|+|y-y'|)$.
For a concrete example, we can consider 
$$
u(x,y)=v ( a x + b y)
$$
for a univariate utility function $v:\R\to\R$ with bounded derivative and $a$ and $b$ are two constants. 
We do not assume other properties of $u$ than Lipschitz continuity, and the infinite-dimensional optimization problem \eqref{eq:R1-opt} may be non-convex, non-monotone, and difficult to solve. 

In practice, the stochastic background risk $Y$ may be subject to uncertainty. Instead, the decision maker can observe some approximations $Y_n$ of $Y$ through  available data and statistical modeling. 
For example, if a simulation mechanism for $Y$ is available, then $Y_n$ may represent the simulated data up to step $n$. Another example is that $Y_n$ is an algorithmic approximation of $Y$ up to certain accuracy. 
We assume that $Y_n$ converges to $Y$ in $w^1$, because the Wasserstein distance is a common metric to quantify uncertainty in the distribution of a stochastic object; see e.g., \cite{BCZ22}.
Note that by \eqref{eq:R1-2}, $Y_n \to Y$ in $L^1$ is stronger than this assumption, which is equivalent to $Y_n\to Y$ in distribution plus uniform integrability. 

Suppose that, with only access to $Y_n\in L^1(\mathcal G^\perp)$, the decision can approximately solve the problem
 \begin{equation}
    \label{eq:R1-opt2}
    \begin{aligned} 
\mbox{maximize}~~~ & \E[u(-X,Y_n)]\\
\mbox{subject to}~~~ & X\in L^1(\mathcal G); ~\rho(X)\le r_0 ; ~P(-X)\le x_0.
\end{aligned}
\end{equation}
For $\epsilon\ge 0$, 
we say that $X^*$ is an $\epsilon$-optimizer of \eqref{eq:R1-opt2} if $\E[u(-X^*,Y_n)]\ge \sup_{X\in \mathcal A} \E[u(-X,Y_n)]-\epsilon$,
where  $\mathcal A =\{X\in L^1(\mathcal G): \rho(X)\le r_0;~P(-X)\le x_0\}$ is the set of all feasible decision variables.
Analogously, we define $\epsilon$-optimizer of \eqref{eq:R1-opt}. 
Our definition of $\epsilon$-optimizers allows for an additive error of $\epsilon$ on the optimal value. 
In practice, it can be much less costly to compute an approximate optimizer than an exact optimizer; see  \cite{V01} and \cite{ACGKMP12} for  general  treatments of approximate optimizers and algorithms. Clearly, $\epsilon=0$ corresponds to true optimizers. Note that for any $\epsilon>0$, an $\epsilon$-optimizer of  \eqref{eq:R1-opt2} exists.

 The next result shows that any sequence of $\epsilon$-optimizers converges to an $\epsilon$-optimizer of the original problem \eqref{eq:R1-opt}, justifying the relevance of using the approximation \eqref{eq:R1-opt2}.

\begin{proposition}\label{prop:invest}
Suppose that $Y_n\to Y$ in $w^1$,  $X_n$ is an $\epsilon_n$-optimizer of \eqref{eq:R1-opt2} with $\epsilon_n\ge 0$ for $n\in \N$,   $u$ is Lipschitz continuous, and $\rho,P\in \mathcal R$. 
Then any subsequence of  $(X_n)_{n\in \N}$ has a cluster point in $w^1$ that is an $\epsilon$-optimizer of \eqref{eq:R1-opt}, where $\epsilon=\limsup_{n\to\infty} \epsilon_n$. In particular, if $\epsilon_n\to 0$, then any convergent (in $w^1$) subsequence  converges to an  optimizer of \eqref{eq:R1-opt}. 
\end{proposition}

\begin{proof}
The constraints in \eqref{eq:R1-opt2} guarantee that $\rho(X_n)$ and $P(-X_n)$ are bounded above.
By Corollary \ref{coro:R1-1}, any subsequence of $(X_n)_{n\in \N}$ has a convergent subsequence. Therefore, it suffices to consider the case $X_n \to X^*$ in $w^1$. Moreover, we can choose $X^*\in L^1(\mathcal G)$ because $w^1$ only concerns the distribution of $X^*$, and $(\Omega,\mathcal G,\p)$ is atomless. 

    By \eqref{eq:R1-2}, we can construct
    $X'_n \in L^1(\mathcal G)$ and $Y'_n\in L^1(\mathcal G^\perp)$ for each $n$, such that  
   $$X_n\laweq X_n';~ Y_n\laweq Y_n';~
   X_n'\to X^* \mbox{~and~} Y_n'\to Y \mbox{~in $L^1$}.$$
   The proof of Corollary \ref{coro:R1-1} justifies 
   $X'_n \to X^*$ almost surely, and also in $L^1$ due to uniform integrability. 
   By  lower semicontinuity and law invariance of $\rho$ and $P$, $$\rho(X^*) \leq \liminf_{n \to \infty} \rho(X'_n) = \liminf_{n \to \infty} \rho(X_n) \leq r_0,$$ and similarly $P(-X^*) \leq x_0$. Hence, $X^*\in \mathcal A$.

    To see that $X^*$ is $\epsilon$-optimal,  let
    $$
    \delta_n = w^1(X_n, X^*) \vee w^1(Y_n, Y) 
    = \E[|X_n'-X^*|] \vee \E[|Y_n'-Y|],
    $$
    which converges to $0$ as $n\to\infty$.
%    $n $ be such that $w^1(X_n, X^*)=\E[|X_n'-X^*|] \leq \delta  $ and $w^1(Y_n, Y) =\E[|Y_n'-Y|] \leq \delta $.   
Using Lipschitz continuity of $f$ repeatedly and 
   the fact that $\E[f(Z,W)]$ for $Z\in L^1(\mathcal G)$
   and $W\in L^1(\mathcal G^\perp)$
   depends only on the laws of $Z$ and $W$, we have 
    \begin{align*}
        \E[f(X^*, Y)] &\ge \E[f(X^*, Y'_n)]
        -\E[c |Y-Y'_n|] 
        \\&\ge  \E[f(X^*, Y'_n)] 
        -c \delta_n \\
        &\ge   \E[f(X_n', Y'_n)]  - \E[c|X^*-X_n'|]
        -c \delta_n   \\
        &\ge   \E[f(X'_n, Y'_n)]    
        -2 c \delta_n   \\
        &=  \E[f(X_n, Y_n)]    
        -2 c \delta_n   \\
        & \ge  \sup_{X\in \mathcal A} \E[f(X, Y_n)]-\epsilon_n -2c \delta_n  \\
           & = \sup_{X\in \mathcal A} \E[f(X, Y'_n)]-\epsilon_n -2c \delta_n  \\
        &\ge  \sup_{X\in \mathcal A}  \E[f(X, Y)] - \E[c|Y'_n-Y|] -\epsilon_n - 2  c \delta_n      \ge  \sup_{X\in \mathcal A} 
 \E[f(X, Y)] -\epsilon_n - 3   c \delta_n.   
    \end{align*} 
Since $\delta_n\to 0$, by taking a limit as $n\to\infty$, we conclude that 
$\E[f(X^*,Y)] \ge \sup_{X\in \mathcal A} 
 \E[f(X, Y)]  -\epsilon$, and thus the desired $\epsilon$-optimality holds. 
\end{proof}

Under the stated conditions in Proposition \ref{prop:invest},  
  a true optimizer of \eqref{eq:R1-opt} always exists, and this can be seen  by setting $Y_n=Y$ for $n\in \N$, and using the fact that $(1/n)$-optimizers exist. 
  Proposition \ref{prop:invest} also implies that all cluster points of $(X_n)_{n\in \N}$ are $\epsilon$-optimizers of \eqref{eq:R1-opt2}.

To interpret Proposition \ref{prop:invest}, the decision maker may use a low-cost algorithm to compute an $\epsilon_n$-optimizer of \eqref{eq:R1-opt2} for some $n$  according to available computational resources and information of $Y_n$. Proposition \ref{prop:invest} guarantees that, by increasing $n$, any convergent subsequence provided by this procedure  converges to an $\epsilon$-optimizer of \eqref{eq:R1-opt}, which the decision maker does not have access to.

In the proof of  Proposition \ref{prop:invest}, our main results in Theorem \ref{th:3-ch} are used, through the conditions  $\rho(X)\le r_0$ and  $P(-X)\le x_0$,  to establish uniform integrability of $(X_n)_{n\in \N}$. This further allows the construction of $(X_n')_{n\in \N}$ which has $L^1$ convergence, guaranteeing $X^*\in \mathcal A$ as well as its optimality. 
Without Theorem \ref{th:3-ch} or the conditions $\rho(X)\le r_0$ and  $P(-X)\le x_0$, the above arguments fail as $(X_n)_{n\in \N}$ may not have any convergent subsequence.  % and the conclusion in  Proposition \ref{prop:invest} may not hold.

\begin{remark}
As we can see from the proof of  Proposition \ref{prop:invest}, 
the advantages brought by  
uniform integrability in optimization rely on the condition that the value function exhibits $L^1$-type continuity properties. In particular, the Lipschitz continuity of $u$ excludes functions  of the form $ u(x,y)=v ( a x + b y) $ where $v$ grows super-linearly. 
Admittedly, these assumptions appear to be quite restrictive, and they capture the additional guarantee  that uniform integrability offers. 
\end{remark}

\section{Conclusion}
\label{sec:6}

The main contribution of this paper is to establish a  connection between boundedness of risk measure values and uniform integrability. As a convenient technical tool, we obtained an upper bound on the folding score of each distortion risk measure (Theorem \ref{th:main}).
Three different sets of equivalent conditions for uniform integrability are obtained,   via ES (Theorem \ref{th:1}),    via distortion risk measures (Theorem \ref{th:3}) and    via law-invariant coherent risk measures (Theorem \ref{th:3-ch}).
Conditions in these results are stated both for  the absolute value of the random variables involved  and for the random variables themselves, facilitating easy interpretation in risk management.   
An application of investment optimization illustrates how our main results and the formulation with random losses instead of their absolute values are helpful to establish convergence of optimal decisions.
These results form a bridge between two important concepts, one in probability theory and one in financial mathematics. Moreover, they highlight the symmetric roles played by the mappings $X\mapsto \rho_h(X)$ 
and $X\mapsto \E[\phi(X)]$, well known in decision theory (\cite{Y87}). A mathematical duality between the distributional transforms underlying these two classes of mappings, in the sense that one class is characterized by commutation with the other, is recently obtained by \cite{CLW23}.

Some remaining questions concern boundedness of the folding score of general law-invariant coherent risk measures, 
and conditions on the ``testing'' risk measure for a strong law of large numbers. 
These two questions are discussed in Remarks  \ref{rem:2} and \ref{rem:1}.

\section*{Acknowledgements}
The authors thank the Editor, the Associate Editor,   two anonymous referees,  and Felix Liebrich for very helpful discussions and comments. 
RW acknowledges financial support from the Natural Sciences and Engineering Research Council of Canada (RGPIN-2018-03823 and CRC-2022-00141).
\section*{Competing Interests}  
The authors declare no competing interests.

\appendix\normalsize
\section{Technical discussions on folding score}
\label{app:folding}

In this appendix, we discuss some technical details related to the folding score and Theorem \ref{th:main} in Section \ref{sec:general}.
\subsection{Sharpness of the bound in Theorem \ref{th:main}}
Let us denote the upper bound in \eqref{eq:thmain} by $$s_{\rho_h}\le b_h = \frac{h(1/2)+1/2}{h(1/2)-1/2}.$$ Note that $b_h$ is decreasing in $h(1/2)$, and its smallest value $3$ is attained when $h(1/2)=1$.
We provide two examples, first showing that the bound $3$ cannot be improved, and   second showing that it is not always sharp for a given $h$.

\begin{example}[The upper bound $3$ is sharp for some $h$]\label{ex:1}
Let $\rho_h=\ES_p$ for $p\in [1/2,1)$. In this case, $h(1/2)=1$, and the upper bound in  \eqref{eq:thmain} is $3$.  
For $\epsilon \in (0,1)$, we can construct $X$ such that $s_{\rho_h}(X)$ is precisely $  3-\epsilon$, and this shows  $s_{\rho_h}=3$. Thus, the bound $3$ cannot be improved  for $\rho_h$.  
Let $X$ have a two-point distribution given by
$\p(X=-1)=1-w$ and $\p(X= 2(1-p)/w)=w$, where $w =  \epsilon({1-p})/({4 -\epsilon}) \in (0,1-p)$. We have $$\ES_p(|X|)=2+(1-p-w)\frac{1}{1-p} = 3-\frac{w}{1-p}=3-\frac{\epsilon}{4-\epsilon},$$ and $$\ES_p(X) \vee \ES_p(-X)= \big( 2-(1-p-w)\frac{1}{1-p}\big) \vee 1 = 1 + \frac{w}{1-p} =1+\frac{\epsilon}{4-\epsilon}. $$
 Therefore, $$
\frac{ \ES_p(|X|)} { \ES_p(X) \vee \ES_p(-X)} =
 \frac{3 (4-\epsilon) -{\epsilon} }{(4-\epsilon) +\epsilon} =   3-\epsilon.$$

\end{example}

\begin{example}[The upper bound in \eqref{eq:thmain} is  not sharp for some $h$]
Let $\rho_h=\ES_p$ for $p\in (0,1/2]$. In this case, $h(1/2)=(2(1-p))^{-1}$, and the upper bound in  \eqref{eq:thmain} is $2/p-1$.
Note that this bound increases to $\infty$ as $p\downarrow 0$. This is intuitive, as $\ES_0$ corresponds to the mean, whose folding score is infinity. 
The bound $2/p-1$ is not sharp. Take $p = 1/4$, and thus $2/p-1 = 7$. However, for any $z \in (0,1)$, $a+b = g(z)/h(z)+g(1-z)/h(1-z) \leq 1$, where $a,b$ are defined in the proof of Theorem \ref{th:main}. Hence, $s_{\rho_h} \leq (2+a+b)/(1-ab) \leq 3/(1-c) = 6$. 
\end{example}

Although the bound in \eqref{eq:thmain} is not always sharp, it suffices for our study in Section \ref{sec:UI}, where we only need $s_\rho<\infty$ as in  Proposition \ref{th:cor1}.

\subsection{Other families of risk measures}

Theorem \ref{th:main} shows  that,  
   for the large class of  coherent distortion risk measures $\rho$, we have 
\begin{align}\label{eq:thmain-2}
   s_{\rho} 
=\sup_{X\in \L}  \frac{\rho(|X|)}{|\rho(X)|\vee |\rho(-X)|}
< \infty \mbox{ unless $\rho=\E$},
   \end{align}
  % Note that \eqref{eq:thmain-2} remains true if $L^1$ is replaced by $L^\infty$, since any $X\in L^1$ can be approximated by sequences in $L^\infty$. 
  One may naturally wonder whether 
\eqref{eq:thmain-2}   holds for other classes of risk measures, that is, whether the folding score is finite for non-linear risk measures. 
With a series of counter examples, we answer the question negatively for many classes of interesting risk measures.
It suffices to construct examples in $ L^\infty$, which we will work with in all examples below.

\begin{example}[Law-invariant convex risk measures]
Property \eqref{eq:thmain-2} does not hold for law-invariant convex risk measures in general, illustrated below.
The entropic risk measures, indexed by parameter $\beta >0$, are defined by
$$
\mathrm{ER}_\beta(X)  = \frac{1}{\beta }\log \E[\exp(\beta X)],~~~~~X\in L^1.
$$
The entropic risk measures are a popular class of law-invariant convex risk measures (\cite{FS16}). 
We will show that this class does not satisfy \eqref{eq:thmain-2}. 
Taking $X_\lambda = \lambda (2\id_A -1)$ where $\lambda>0$ and $A\in \mathcal F$ with $\p(A)=1/2$, we have 
$
\mathrm{ER}_\beta(|X_\lambda|) = \lambda 
$
and 
$$
\mathrm{ER}_\beta(X_\lambda) = \mathrm{ER}_\beta(-X_\lambda) =\frac{1}{\beta }\log \big(\frac 12 \exp(\beta\lambda) +\frac 12 \exp (-\beta\lambda)\big)\ge 0.
$$
Therefore, we have 
$$
\lim_{\lambda \downarrow 0}  \frac{\mathrm{ER}_\beta(X_\lambda) }{\mathrm{ER}_\beta(|X_\lambda|) } = \lim_{\lambda \downarrow 0} \frac{\frac{1}{\beta }\log (\frac 12 \exp(\beta\lambda) +\frac 12 \exp (-\beta\lambda))}{\lambda}
= \lim_{\lambda \downarrow 0} \frac{\log (\frac 12 \exp( \lambda) +\frac 12 \exp (- \lambda))}{\lambda}=0,
$$
where we used the l'Hospital rule and 
$$
\frac{\d }{\d \lambda}\log \big(\frac 12 \exp( \lambda) +\frac 12 \exp (- \lambda)\big) 
= \frac{\frac 12 \exp( \lambda) -\frac 12 \exp (- \lambda)}{\frac 12 \exp( \lambda) +\frac 12 \exp (- \lambda)} \to 0 \mbox{~as $\lambda\downarrow 0$}.
$$  
Therefore, $s_{\rho}=\infty$ for $\rho = \mathrm{ER}_\beta$.
\end{example}

A risk measure $\rho$ on $L^\infty$ is Fatou continuous if $\liminf_{n\to\infty} \rho(X_n)\ge \rho(X)$ for all uniformly bounded sequences $(X_n)_{n\in \N}$ with $X_n\to X$ in probability. Note that Fatou continuity is slightly weaker than $L^1$-lower semicontinuity (confined to $\mathcal L=L^\infty$) due to the type of convergence; see \citet[Chapter 7]{R13}.
\begin{example}[Coherent risk measures]
Property \eqref{eq:thmain-2} does not generalise to the class of Fatou-continuous coherent risk measures on $L^\infty$.
 This class of risk measures can be represented by  $$
  \rho = \sup_{Q\in \mathcal Q} \E^Q,
  $$
  where $\mathcal Q$ is a set of probability measures absolutely continuous with respect to $\p$.
Consider $[0,1]$ with the  Lebesgue measure $\lambda$.  Let $Q_1$ and $Q_2$ be defined by their Radon-Nikodym derivatives $$
\frac{\d Q_1 }{ \d \lambda}  = \frac 34\id_{[0,1/3]\cup [2/3,1]} + \frac 32\id_{(1/3,2/3)} \mbox{~~and~~}\frac{\d Q_2 }{ \d \lambda}  = \frac 32 \id_{[0,1/3]\cup [2/3,1]}.$$  
Let $X = \id_{[0,1/3]} - \id_{[2/3,1]} , \mathcal{Q} = \{Q_1, Q_2\},$ and $ \rho = \sup_{Q\in \mathcal Q} \E^Q.$ Then $s_\rho(X) = 1/0 = \infty$. 
\end{example}

\begin{example}[Coherent Choquet risk measures]
Property \eqref{eq:thmain-2} does not hold for the class of coherent Choquet risk measures.
A coherent Choquet risk measure has the form
 $$
   \rho_\nu(X) = \int_0^\infty  \nu (X>x)\d x + \int_{-\infty}^0 \big(\nu(X>x)-1\big)\d x, ~~~X\in L^\infty,
   $$
   where $\nu:\mathcal F\to [0,1]$  is increasing  and satisfies $\nu(N)=0$ for $N\in \mathcal F$ with $0$ probability, $\nu(\Omega)=1$,
   and $\nu(A\cup B) + \nu (A\cap B) \le \nu(A) +\nu(B)$ for all $A,B\in \mathcal F$.  If $\nu=h\circ \p$ for a concave $h\in \mathcal D$, then $\rho_\nu =\rho_h$. 
Consider $S = [-1/3, 1/3]\times [-1,1]$ with Lebesgue measure $\mu$.  Denote $S^+ = [-1/3, 1/3]\times [0,1]$ and $S^- = [-1/3, 1/3]\times [-1,0)$. Define $\nu(A) := (\mu(A \cap S^+)\wedge 1/2) + (\mu(A \cap S^-)\wedge 1/2)$. Take $X = \id_{S^+} - \id_{S^-}$. Then $s_{\rho_\nu}(X) = 1/0 = \infty$.
\end{example}

\begin{remark}\label{rem:1}
It remains an open question whether property \eqref{eq:thmain-2} holds for the class of law-invariant coherent risk measures. For this class, we did not find any example  of $\rho \ne \E$ satisfying $s_\rho=\infty$, 
although we could not prove $s_\rho<\infty$ for all $\rho\ne \E$.
By Kusuoka's representation (\cite{K01}),
any law-invariant coherent risk measure $\rho$ on $L^\infty$ can be represented as
$
\rho = \sup_{h \in \mathcal H_\rho } \rho_h
 $ 
where $\mathcal H_\rho$ is a set of concave distortion functions. 
If $\mathcal H_\rho $ is a finite set, 
then by using Theorem \ref{th:main} we can obtain \eqref{eq:thmain-2}.
For an infinite $\mathcal H_\rho$, this is not clear. 
\end{remark}

\end{document}